\documentclass[a4paper,10pt,]{article}
\usepackage[utf8,]{inputenc}

\usepackage{setspace}
\usepackage{amsmath,amsfonts,amssymb}

\newtheorem{theorem}{Theorem}

\newtheorem{definition}[theorem]{Definition}
\newtheorem{remark}[theorem]{Remark}
\newtheorem{example}[theorem]{Example}
\newtheorem{corollary}[theorem]{Corollary}
\newtheorem{proposition}[theorem]{Proposition}
\newenvironment{proof}[1][Proof]{\noindent\textbf{#1: }}{\ \rule{0.5em}{0.5em}}

\title{Continuous set packing problems\\ and near-Boolean functions}
\author{Giovanni Rossi\\
\footnotesize{Departments of Computer Science and Engineering - DISI, and Mathematics}\\
\footnotesize{University of Bologna, Mura Anteo Zamboni 7, Italy 40126, \textsl{giovanni.rossi6@unibo.it}}}

\begin{document}

\maketitle

\begin{abstract}
Given a family of feasible subsets of a ground set, the packing problem is to find a largest subfamily of pairwise disjoint family members. Non-approximability
renders heuristics attractive viable options, while efficient methods with worst-case guarantee are a key concern in computational complexity. This work proposes a
novel near-Boolean optimization method relying on a polynomial multilinear form with variables ranging each in a high-dimensional unit simplex, rather than in
the unit interval as usual. The variables are the elements of the ground set, and distribute each a unit membership over those feasible subsets where they are included.
The given problem is thus translated into a continuous version where the objective is to maximize a function taking values on collections of points in a unit hypercube.
Maximizers are shown to always include collections of hypercube disjoint vertices, i.e. partitions of the ground set, which constitute feasible solutions for the
original discrete version of the problem. A gradient-based local search in the expanded continuous domain is designed. Approximations with polynomial multilinear form of
bounded degree and near-Boolean games are also finally discussed.\smallskip

\smallskip

\end{abstract}

\section{Introduction}
Consider a finite set $N=\{1,\ldots ,n\}$ of items to be packed into feasible subsets, where these latter constitute a family $\mathcal F\subseteq 2^N=\{A:A\subseteq N\}$.
The problem is to find a subfamily $\mathcal F^*\subseteq\mathcal F$ of pairwise disjoint feasible subsets with largest size $|\mathcal F^*|$. In the weighted version,
a function $w:\mathcal F\rightarrow\mathbb R_+$ identifies as optimal those such subfamilies $\mathcal F^*$ with maximum weight $W(\mathcal F^*)=\sum_{A\in\mathcal F^*}w(A)$.
Maximizing $|\mathcal F^*|$ is equivalent to setting $w(A)=1$ for all $A\in\mathcal F$. Thus this work proposes to use the polynomial multilinear extension, or MLE
for short, of set functions (such as $w$) in order to evaluate families of fuzzy feasible subsets. Although unfeasible, such families shall still drive the search towards
locally optimal feasible ones.

Set packing is a key combinatorial optimization problem \cite{KorteVygen2002} extensively studied in computational complexity, where the aim is to find efficient algorithms
whose output approximates optimal solutions within a provable bounded factor. In that field, the focus is placed mostly on non-approximability results for $k$-set packing,
where the size of every family member is no greater than some $k\ll n$ (and with unit weight for each member as above, \cite{Trevisan01}). Recall that if all family members
have size $k=2$, then the problem is to find a maximal matching in a graph with vertex set $N$, and an efficient (i.e. with polynomial running time) algorithm capable
to output an exact solution is known to exist \cite{Papadimitriou94}. In fact, if $k>2$, then $k$-set packing may be rephrased in terms of vertex clouring in hypergraphs,
with special focus on the $d$-regular and $k$-uniform case, where every element of the ground set is present in precisely $d>1$ feasible subsets
(i.e. $|\{A:i\in A\in\mathcal F\}|=d$ for every $i\in N$), each of which, in turn, has size $k$ (i.e. $|A|=k$ for every $A\in\mathcal F$) \cite{Hazan++06complexity}.

Set packing aslo has important applications, among which combinatorial auctions constitute a main and lucrative example: the ground set may consist of items to be
sold in bundles (or subsets) towards revenue maximization, and once bids are processed the issue may be tackled as a maximum-weight set packing problem, with maximum received
bids on bundles as weights \cite{Sandholm02}. Given the exponentially large size of the search space, revenue maximization often leads to use heuristics with no worst-case
guarantee or, more simply, to sell each item independently but simultaneously over a sufficiently long time period \cite{Milgrom04}.

The approach to set packing problems proposed in the sequel replaces standard pseudo-Boolean optimization \cite{BorosHammer02} with a novel near-Boolean method. While the former
employs MLE to switch from $\{0,1\}$ to $[0,1]$ as the domain of each of the $n$ variables, the proposed near-Boolean method relies on $n$ variables ranging each in the $2^{n-1}$-set
of extreme points of a unit simplex, and employs MLE to include the continuum provided by the whole simplex. The $n$ variables correspond to the elements $i\in N$ of the ground set,
while the extreme points of each simplex are indexed by those subsets where each element is included. Then, the MLE of the resulting near-Boolean function allows to evaluate collections
of fuzzy subsets of $N$ or, equivalently, fuzzy subfamilies of feasible subsets $A\in\mathcal F$. The objective function to be maximized takes thus values over the $n$-product of
$2^{n-1}-1$-dimensional unit simplices, allowing to design a flexible gradient-based local search.

The following section comprenshively details the framework for the full-dimensional case $\mathcal F=2^N$. This not only seems useful for applications, but most importantly allows to
clearly see next that by simply introducing the empty set $\emptyset$ and all $n$ singletons $\{i\}\in 2^N$ into the family $\mathcal F$ of feasible subsets (with null weights
$w(\emptyset)=0=w(\{i\}$ if $\{i\}\notin\mathcal F$) the whole class of set packing problems may be handled by the proposed method. The gradient-based local search differs when switching
from the full-dimensional case to the lower-dimensional one $\mathcal F\subset 2^N$, in that with the latter a cost function $c:\mathcal F\rightarrow\mathbb N$ also enters the picture,
in line with greedy approaches to weighted set packing \cite{Chandra+2001}. The cost $c(A)=|\{B:B\in\mathcal F,B\cap A\neq\emptyset\}|$ of including a feasible subset in the packing
is the number of members with which it has non-empty intersection (itself included, hence $c(A)\in\mathbb N,A\in\mathcal F$).

Note that maximum-weight set packing may be tackled through constrained maximization of standard pseudo-Boolean function $v:\{0,1\}^{|\mathcal F|}\rightarrow\mathbb R_+$
\begin{equation*}
\text{given by }v\left(x_{A_1},\ldots ,x_{A_{|\mathcal F|}}\right)=\sum_{1\leq k\leq |\mathcal F|}x_{A_k}w(A_k)\text{ , s.t.}
\end{equation*}
\begin{equation*}
A_k\cap A_l\neq\emptyset\Rightarrow x_{A_k}+x_{A_l}\leq 1\text{ for all }1\leq k<l\leq|\mathcal F|\text ,
\end{equation*}
where $x_A\in\{0,1\}$ for all $A\in\mathcal F=\{A_1,\ldots A_{|\mathcal F|}\}$. Also, $v$ can be replaced with $xMx\simeq v$, where $M$ is a suitable
$|\mathcal F|\times|\mathcal F|$-matrix and $x=(x_{A_1},\ldots ,x_{A_{|\mathcal F|}})$ \cite{AlidaeeEtAl2008}. An heuristic then finds a constrained maximizer $x^*$, while
the corresponding solution is $\mathcal F^*=\{A:x_A^*=1\}$. This differs from what is proposed here, in many respects, the most evident of which being that $v$
has $|\mathcal F|$ constrained Boolean variables, while the expanded MLE developed below has $n$ unconstrained near-Boolean variables. 

\section{Full-dimensional case}
The $2^n$-set $\{0,1\}^n$ of vertices of the $n$-dimensional unit hypercube $[0,1]^n$ corresponds one-to-one to the (power) set $2^N$ of subsets $A\subseteq N$ through
characteristic functions $\chi_A:N\rightarrow\{0,1\},A\in 2^N$ defined by $\chi_A(i)=1$ if $i\in A$ and $\chi_A(i)=0$ if $i\in N\backslash A=A^c$, while
collection $\{\zeta(A,\cdot):A\in 2^N\}$ is a linear basis of the vector space $\mathbb R^{2^n}$ of real-valued functions $w$ on $2^N$, where zeta function
$\zeta:2^N\times2^N\rightarrow\mathbb R$ is the element of the incidence algebra \cite{Rota64,Aigner79} of Boolean lattice $(2^N,\cap,\cup)$ defined by $\zeta(A,B)=1$ if
$B\supseteq A$ and $\zeta(A,B)=0$ if $B\not\supseteq A$. Linear combination $w(B)=\sum_{A\in 2^N}\mu^w(A)\zeta(A,B)=\sum_{A\subseteq B}\mu^w(A)$ for $B\in 2^N$ applies to
any $w$, with M\"obius inversion $\mu^w:2^N\rightarrow\mathbb R$ uniquely given by ($\subset$ is strict inclusion) $\mu^w(A)=$
\begin{equation*}
=\sum_{B\subseteq A}(-1)^{|A\backslash B|}w(B)\text{ }\left(\text{with }\zeta(B,A)=(-1)^{|A\backslash B|}\right)
\end{equation*}
\begin{equation*}
=w(A)-\sum_{B\subset A}\mu^w(B)\text{ }\left(\text{recursion, with }w(\emptyset)=0\right)\text .
\end{equation*}
Given this essential combinatorial \textit{\textquotedblleft analog of the fundamental theorem of the calculus\textquotedblright} \cite{Rota64}, the MLE
$f^w:[0,1]^n\rightarrow\mathbb R$ of $w$ takes values $w(B)=$
\begin{equation*}
=f^w(\chi_B)=\sum_{A\in 2^N}\left(\prod_{i\in A}\chi_B(i)\right)\mu^w(A)=\sum_{A\subseteq B}\mu^w(A)
\end{equation*}
\begin{equation}
\text{on vertices, and }f^w(q)=\sum_{A\in 2^N}\left(\prod_{i\in A}q_i\right)\mu^w(A)
\end{equation}
on any point $q=(q_1,\ldots ,q_n)\in[0,1]^n$. Conventionally, $\prod_{i\in\emptyset}q_i:=1$ \cite[p. 157]{BorosHammer02}.

Let $2^N_i=\{A:i\in A\in 2^N\}=\left\{A_1,\ldots ,A_{2^{n-1}}\right\}$ be the $2^{n-1}$-set of subsets containing each $i\in N$. Simplex
\begin{equation*}
\Delta_i=\left\{\left(q_i^{A_1},\ldots ,q_i^{A_{2^{n-1}}}\right)\in\mathbb R^{2^{n-1}}_+: \sum_{1\leq k\leq 2^{n-1}}q_i^{A_k}=1\right\}
\end{equation*}
has dimension $2^{n-1}-1$ and generic point $q_i\in\Delta_i$.
\begin{definition}
A fuzzy cover $\textbf q$ specifies a membership distribution for each $i\in N$ over the $2^{n-1}$ subsets containing it, i.e. 
$\textbf q=(q_1,\ldots ,q_n)\in\Delta_N=\times_{1\leq i\leq n}\text{ }\Delta_i$.
\end{definition}
Equivalently, $\textbf q=\left\{q^A:\emptyset\neq A\in 2^N,q^A\in[0,1]^n\right\}$ is a $2^n-1$-set whose elements $q^A=\left(q^A_1,\ldots ,q^A_n\right)$ are $n$-vectors
corresponding to non-empty subsets $A\in 2^N$ and specifying a membership $q_i^A$ for each $i\in N$, with $q_i^A\in[0,1]$ if $i\in A$ while $q_i^A=0$ if
$i\in A^c$. Fuzzy covers being collections of points in $[0,1]^n$, and the MLE $f^w$ of $w$ allowing precisely to evaluate such points, the global worth $W(\textbf q)$ of
$\textbf q\in\Delta_N$ is the sum over all $q^A,A\in 2^N$ of $f^w(q^A)$ as defined by (1). That is,
\begin{equation*}
W(\textbf q)=\sum_{A\in 2^N}f^w(q^A)=\sum_{A\in 2^N}\left[\sum_{B\subseteq A}\left(\prod_{i\in B}q_i^A\right)\mu^w(B)\right]\text ,
\end{equation*}
\begin{equation}
\text{or }W(\textbf q)=\sum_{A\in 2^N}\left[\sum_{B\supseteq A}\left(\prod_{i\in A}q_i^B\right)\right ]\mu^w(A)\text .
\end{equation}
\begin{example}
For $N=\{1,2,3\}$, define $w(\{1\})=w(\{2\})=w(\{3\})=0.2$, $w(\{1,2\})=0.8$, $w(\{1,3\})=0.3$, $w(\{2,3\})=0.6$, $w(N)=0.7$. Membership distributions of elements
$i=1,2,3$ over $2^N_i$ are $q_1\in\Delta_1,q_2\in\Delta_2,q_3\in\Delta_3$,
\begin{equation*}
q_1=\left(\begin{array}{c}q_1^1 \\q_1^{12} \\q_1^{13} \\q_1^N\end{array}\right)\text{, }
q_2=\left(\begin{array}{c}q_2^2 \\q_2^{12} \\q_2^{23} \\q_2^N\end{array}\right)\text{, } 
q_3=\left(\begin{array}{c}q_3^3 \\q_3^{13} \\q_3^{23} \\q_3^N\end{array}\right)\text .
\end{equation*}
If $\hat q_1^{12}=\hat q_2^{12}=1$, then any membership $q_3\in\Delta_3$ yields
\begin{eqnarray*}
W(\hat q_1,\hat q_2,q_3)&=&w(\{1,2\})\\
&+&\left(q_3^3+q_3^{13}+q_3^{23}+q_3^N\right)\mu^w(\{3\})\\
&=&w(\{1,2\})+w(\{3\})=1\text .
\end{eqnarray*}
This means that there is a continuum of fuzzy covers achieving maximum worth, i.e. $1$. In order to select the one 
$\hat{\textbf q}=(\hat q_1,\hat q_2,\hat q_3)$ where $\hat q_3^3=1$, attention must be placed only on \textit{exact} ones, defined hereafter.
\end{example}
For any two fuzzy covers $\textbf q=\{q^A:\emptyset\neq A\in 2^N\}$ and $\hat{\textbf q}=\{\hat q^A:\emptyset\neq A\in 2^N\}$, define $\hat{\textbf q}$ to be a \textit{shrinking} of
$\textbf q$ if there is a subset $A$, with $\sum_{i\in A}q_i^A>0$ and
\begin{eqnarray*}
\hat q^B_i&=&\left\{\begin{array}{c} q^B_i\text{ if } B\not\subseteq A\\ 0\text{ if }B=A\end{array}\right.\text{ for all }B\in 2^N,i\in N\text ,\\
\sum_{B\subset A}\hat q^B_i&=&q^A_i+\sum_{B\subset A}q_i^B\text{ for all }i\in A\text{.}
\end{eqnarray*}
In words, a shrinking reallocates the whole membership mass $\sum_{i\in A}q_i^A>0$ from $A\in 2^N$ to all proper subsets $B\subset A$, involving all and only
those elements $i\in A$ with strictly positive membership $q_i^A>0$.
\begin{definition}
Fuzzy cover $\textbf q\in\Delta_N$ is exact as long as $W(\textbf q)\neq W(\hat{\textbf q})$ for all shrinkings $\hat{\textbf q}$ of \textbf q.
\end{definition}
\begin{proposition}
If $\textbf q$ is exact, then $\left|\left\{i\in A:q_i^A>0\right\}\right|\in\{0,|A|\}$ for all $A\in 2^N$.
\end{proposition}
\begin{proof}
For $\emptyset\subset A^+(\textbf q)=\left\{i:q_i^A>0\right\}\subset A$, with $\alpha=|A^+(\textbf q)|>1$, note that\bigskip

$f^w(q^A)=\sum_{B\subseteq A^+(\textbf q)}\left(\prod_{i\in B}q_i^A\right)\mu^w(B)$.\bigskip

Let shrinking $\hat{\textbf q}$, with $\hat q^{B'}=q^{B'}$ if $B'\not\in 2^{A^+(\textbf q)}$, satisfy conditions\bigskip

1) $\sum_{B\in 2^N_i\cap 2^{A^+(\textbf q)}}\hat q_i^B=q_i^A+\sum_{B\in 2^N_i\cap 2^{A^+(\textbf q)}}q_i^B\text{ for all }i\in A^+(\textbf q)$, and\bigskip

2) $\prod_{i\in B}\hat q_i^B=\prod_{i\in B}q_i^B+\prod_{i\in B}q_i^A\text{ for all }B\in 2^{A^+(\textbf q)}$ such that $|B|>1$.\bigskip

These are $2^\alpha-1$ equations with $\sum_{1\leq k\leq\alpha}k\binom{\alpha}{k}>2^{\alpha}$ variables $\hat q_i^B,B\subseteq A^+(\textbf q)$, $i\in B$. Thus there is a
continuum of solutions, each providing precisely a shrinking $\hat{\textbf q}$ where\bigskip

$\sum_{B\in 2^{A^+(\textbf q)}}f^w(\hat q^B)=f^w(q^A)+\sum_{B\in 2^{A^+(\textbf q)}}f^w(q^B)$.\bigskip

This entails that \textbf q is not exact.
\end{proof}

Partitions $P=\{A_1,\ldots ,A_{|P|}\}\subset 2^N$ of $N$ are families of pairwise disjoint subsets called blocks \cite{Aigner79}, that is
$A_k\cap A_l=\emptyset,1\leq k<l\leq |P|$, with union $N=\cup_{1\leq k\leq |P|}\text{ }A_k$. Any $P$ corresponds to the collection $\{\chi_A:A\in P\}$ of those
$|P|$ hypercube vertices identified by the characteristic functions of its blocks (see above). Partitions $P$ can also be seen as $\textbf p\in\Delta_N$
where $p_i^A=1$ for all $A\in P,i\in A$, i.e. exact fuzzy covers where each $i\in N$ concentrates its whole membershisp on a unique $A\in 2^N_i$, thus justifying
the following.
\begin{definition}
Fuzzy partitions are exact fuzzy covers.
\end{definition}
Ojective function $W:\Delta_N\rightarrow\mathbb R$ includes among its extremizers (non-fuzzy) partitions. This expands a basic result in pseudo-Boolean optimization.
Denote by $ex(\Delta_i)$ the $2^{n-1}$-set of extreme points of $\Delta_i$. For $\textbf q\in\Delta_N,i\in N$, let $\textbf q=q_i|\textbf q_{-i}$, with
$q_i\in\Delta_i$ and $\textbf q_{-i}\in\Delta_{N\backslash i}=\times_{j\in N\backslash i}\text{ }\Delta_j$. Then, for any $w$,
\begin{equation*}
W(\textbf q)=\sum_{A\in 2^N_i}f^w(q^A)+\sum_{A'\in 2^N\backslash 2^N_i}f^w(q^{A'})=
\end{equation*}
\begin{eqnarray*}
&=&\sum_{A\in2^N_i}\sum_{B\subseteq A\backslash i}\left(\prod_{j\in B}q^A_j\right)\Big(q^A_i\mu^w(B\cup i)+\mu^w(B)\Big)\\
&+&\sum_{A'\in 2^N\backslash 2^N_i}\sum_{B'\subseteq A'}\left(\prod_{j'\in B'}q^{A'}_{j'}\right)\mu^w(B')
\end{eqnarray*}
at all $\textbf q\in\Delta_N$ and for all $i\in N$. Now define
\begin{equation*}
W_i(q_i|\textbf q_{-i})=\sum_{A\in 2^N_i}q^A_i\left[\sum_{B\subseteq A\backslash i}\left(\prod_{j\in B}q^A_j\right)\mu^w(B\cup i)\right]\text ,
\end{equation*}
\begin{eqnarray*}
W_{-i}(\textbf q_{-i})&=&\sum_{A\in 2^N_i}\left[\sum_{B\subseteq A\backslash i}\left(\prod_{j\in B}q^A_j\right)\mu^w(B)\right]+\\
&+&\sum_{A'\in 2^N\backslash 2^N_i}\left[\sum_{B'\subseteq A'}\left(\prod_{j'\in B'}q^{A'}_{j'}\right)\mu^w(B')\right]\text ,
\end{eqnarray*}
\begin{equation}
\text{yielding }W(\textbf q)=W_i(q_i|\textbf q_{-i})+W_{-i}(\textbf q_{-i})\text{.}
\end{equation}
\begin{proposition}
For all $\textbf q\in\Delta_N$, there are $\underline{\textbf q},\overline{\textbf q}\in\Delta_N$
\begin{equation*}
\text{such that }\left\{\begin{array}{c}
\text{(i) }W(\underline{\textbf q})\leq W(\textbf q)\leq W(\overline{\textbf q})\text{ and,}\\
\text{(ii) }\underline q_i,\overline q^i\in ex(\Delta_i)\text{ for all }i\in N\text .
\end{array}\right.
\end{equation*}
\end{proposition}
\begin{proof}
For $i\in N,\textbf q_{-i}\in\Delta_{N\backslash i}$, define
$w_{\textbf q_{-i}}:2^N_i\rightarrow\mathbb R$ by
\begin{equation}
w_{\textbf q_{-i}}(A)=\sum_{B\subseteq A\backslash i}\left(\prod_{j\in B}q^A_j\right)\mu^w(B\cup i)\text .
\end{equation}
Let $\mathbb A^+_{\textbf q_{-i}}=\arg\max w_{\textbf q_{-i}}$ and $\mathbb A^-_{\textbf q_{-i}}=\arg\min w_{\textbf q_{-i}}$,
with $\mathbb A^+_{\textbf q_{-i}}\neq\emptyset\neq\mathbb A^-_{\textbf q_{-i}}$ at all $\textbf q_{-i}$. Most importantly,
\begin{equation}
W_i(q_i|\textbf q_i)=\sum_{A\in 2^N_i}\Big(q^A_i\cdot w_{\textbf q_{-i}}(A)\Big)=\langle q_i,w_{\textbf q_{-i}}\rangle\text ,
\end{equation}
where $\langle\cdot ,\cdot\rangle$ denotes scalar product. Thus for given membership distributions of all $j\in N\backslash i$, global worth is affected by $i$'s
membership distribution through a scalar product. In order to maximize (or minimize) $W$ by suitably choosing $q_i$ for given $\textbf q_{-i}$, the whole of $i$'s membership mass
must be placed over $\mathbb A^+_{\textbf q_{-i}}$ (or $\mathbb A^-_{\textbf q_{-i}}$), anyhow. Hence there are precisely $|\mathbb A^+_{\textbf q_{-i}}|>0$
(or $|\mathbb A^-_{\textbf q_{-i}}|>0$) available extreme points of $\Delta_i$. The following procedure selects (arbitrarily) one of them.\smallskip

\textsc{RoundUp}$(w,\textbf q)$\smallskip

\textsl{Initialize:} Set $t=0$ and $\textbf q(0)=\textbf q$.\smallskip

\textsl{Loop:} While there is a $i\in N$ with $q_i(t)\not\in ex(\Delta_i)$,\smallskip

set $t=t+1$ and:
\begin{enumerate}
\item[(a)] select some $A^*\in\mathbb A^+_{\textbf q_{-i}(t)}$,
\item[(b)] define, for all $j\in N,A\in 2^N$,
\begin{equation*}
q^A_j(t)=
\left\{\begin{array}{c}q_j^A(t-1)\text{ if }j\neq i \\1\text{ if }j=i\text{ and } A=A^* \\0\text{ otherwise}\end{array}\right.\text .
\end{equation*}
\end{enumerate}

\textsl{Output:} Set $\overline{\textbf q}=\textbf q(t)$.\smallskip

Every change $q_i^A(t-1)\neq q_i^A(t)=1$ (for any $i\in N,A\in 2^N_i$) induces a non-decreasing variation
$W(\textbf q(t))-W(\textbf q(t-1))\geq 0$. Hence, the sought $\overline{\textbf q}$ is provided in at most $n$ iterations.
Analogously, replacing $\mathbb A^+_{\textbf q_{-i}}$ with $\mathbb A^-_{\textbf q_{-i}}$ yields the sought minimizer $\underline{\textbf q}$
(see also \cite[p. 163]{BorosHammer02}).
\end{proof}
\begin{remark}
For $i\in N,A\in 2^N_i$, if all $j\in A\backslash i\neq\emptyset$ satisfy $q_j^A=1$, then (4) yields $w_{\textbf q_{-i}}(A)=w(A)-w(A\backslash i)$, while 
$w_{\textbf q_{-i}}(\{i\})=w(\{i\})$ regardless of $\textbf q_{-i}$.
\end{remark}
\begin{corollary}
Some partition $P$ satisfies $W(\textbf p)\geq W(\textbf q)$ for all $\textbf q\in\Delta_N$, with $W(\textbf p)=\sum_{A\in P}w(A)$.
\end{corollary}
\begin{proof}
Follows from propositions 4 and 6 (with the above notation associating $\textbf p\in\Delta_N$ to partition $P$).  
\end{proof}

Defining global maximizers is clearly immediate.
\begin{definition}
Fuzzy partition $\hat{\textbf q}\in\Delta_N$ is a global maximizer if $W(\hat{\textbf q})\geq W(\textbf q)$ for all $\textbf q\in\Delta_N$. 
\end{definition}
Concerning local maximizers, consider a vector $\omega=(\omega_1,\ldots ,\omega_n)\in\mathbb R^n_{++}$ of strictly positive weights, with $\omega_N=\sum_{j\in N}\omega_j$,
and focus on the (Nash) equilibrium \cite{Micro} of the game with elements $i\in N$ as players, each strategically choosing its membership
distribution $q_i\in\Delta_i$ while being rewarded with fraction $\frac{\omega_i}{\omega_N}W(q_1,\ldots ,q_n)$ of the global worth attained at any
$(q_1,\ldots ,q_n)=\textbf q\in\Delta_N$.
\begin{definition}
Fuzzy partition $\hat{\textbf q}\in\Delta_N$ is a local maximizer if for all $q_i\in\Delta_i$
and all $i\in N$ inequality $W_i(\hat q_i|\hat{\textbf q}_{-i})\geq W_i(q_i|\hat{\textbf q}_{-i})$ holds (see (3)).
\end{definition}
This definition of local maximizer entails that the \textit{neighborhood} $\mathcal N(\textbf q)\subset\Delta_N$ of any $\textbf q\in\Delta_N$ is
\begin{equation*}
\mathcal N(\textbf q)=\underset{i\in N}{\bigcup}\Big\{\tilde{\textbf q}:\tilde{\textbf q}=\tilde q_i|\textbf q_{-i},\tilde q_i\in\Delta_i\Big\}\text .
\end{equation*}
\begin{definition}
The $(i,A)$-derivative of $W$ at $\textbf q\in\Delta_N$ is
\begin{equation*}
\partial W(\textbf q)/\partial q^A_i=W(\overline{\textbf q}(i,A))-W(\underline{\textbf q}(i,A))=
\end{equation*}
\begin{equation*}
=W_i\Big(\overline q_i(i,A)|\overline{\textbf q}_{-i}(i,A)\Big)-W_i\Big(\underline q_i(i,A)|\underline{\textbf q}_{-i}(i,A)\Big)\text ,
\end{equation*}
with $\overline{\textbf q}(i,A)=\Big(\overline q_1(i,A),\ldots ,\overline q_n(i,A)\Big)$ given by
\begin{equation*}
\overline q_j^B(i,A)=\left\{\begin{array}{c} q_j^B\text{ for all }j\in N\backslash i,B\in 2^N_j \\ 1\text{ for }j=i,B=A\\0\text{ for }j=i,B\neq A \end{array}\right.
\text ,
\end{equation*}
and $\underline{\textbf q}(i,A)=\Big(\underline q_1(i,A),\ldots ,\underline q_n(i,A)\Big)$ given by
\begin{equation*}
\underline q_j^B(i,A)=\left\{\begin{array}{c} q_j^B\text{ for all }j\in N\backslash i,B\in 2^N_j \\ 0\text{ for }j=i\text{ and all }B\in 2^N_i \end{array}\right.
\text ,
\end{equation*}
thus $\nabla W(\textbf q)=\{\partial W(\textbf q)/\partial q^A_i:i\in N,A\in 2^N_i\}\in\mathbb R^{n2^{n-1}}$ is the (full) gradient of $W$ at $\textbf q$.
The $i$-gradient $\nabla_iW(\textbf q)\in\mathbb R^{2^{n-1}}$ of $W$ at $\textbf q=q_i|\textbf q_{-i}$ is set function $\nabla_iW(\textbf q):2^N_i\rightarrow\mathbb R$ defined by
$\nabla_iW(\textbf q)(A)=w_{\textbf q_{-i}}(A)$ for all $A\in 2^N_i$, where $w_{\textbf q_{-i}}$ is given by (4).
\end{definition}
\begin{remark}
Membership distribution $\underline q_i(i,A)$ is the null one: its $2^{n-1}$ entries are all 0, hence $\underline q_i(i,A)\not\in\Delta_i$.
\end{remark}
The setting obtained thus far allows to conceive searching for a local maximizer partition $\textbf p^*$ from given fuzzy partition \textbf q as initial
candidate solution, and while maintaing the whole search within the continuum of fuzzy partitions. This idea may be specified in alternative ways
yielding different local search methods. One possibility is the following.\smallskip

\textsc{LocalSearch}$(w,\textbf q)$\smallskip

\textsl{Initialize:} Set $t=0$ and $\textbf q(0)=\textbf q$, with requirement $|\{i:q_i^A>0\}|\in\{0,|A|\}$ for all $A\in 2^N$.\smallskip

\textsl{Loop 1:} While $0<\sum_{i\in A}q^A_i(t)<|A|$ for a $A\in 2^N$, set $t=t+1$ and
\begin{enumerate}
\item[(a)] select a $A^*(t)\in 2^N$ such that
\begin{equation*}
\sum_{i\in A^*(t)}w_{\textbf q_{-i}(t-1)}(A^*(t))\geq\sum_{j\in B}w_{\textbf q_{-j}(t-1)}(B)
\end{equation*}
for all $B\in 2^N$ such that $0<\sum_{i\in B}q^B_j(t)<|B|$,
\item[(b)] for $i\in A^*(t)$ and $A\in 2^N_i$, define
\begin{equation*}
q_i^A(t)=
\left\{\begin{array}{c}1\text{ if }A=A^*(t)\text ,\\
0\text{ if }A\neq A^*(t)\text ,\end{array}\right.
\end{equation*}
\item[(c)] for $j\in N\backslash A^*(t)$ and $A\in 2^N_j$ with $A\cap A^*(t)=\emptyset$, define $q^A_j(t)=q_j^A(t-1)+$
\begin{equation*}
+\left(w(A)\sum_{\underset{B\cap A^*(t)\neq\emptyset}{B\in 2^N_j}}q_j^B(t-1)\right)
\left(\sum_{\underset{B'\cap A^*(t)=\emptyset}{B'\in 2^N_j}}w(B')\right)^{-1}
\end{equation*}
\item[(d)] for $j\in N\backslash A^*(t)$ and $A\in 2^N_j$ with $A\cap A^*(t)\neq\emptyset$, define
\begin{equation*}
q^A_j(t)=0\text .
\end{equation*}
\end{enumerate}
\textsl{Loop 2:} While $q_i^A(t)=1,|A|>1$ for a $i\in N$ and $w(A)<w(\{i\})+w(A\backslash i)$, set $t=t+1$ and define:
\begin{eqnarray*}
q^{\hat A}_i(t)&=&\left\{\begin{array}{c}1\text{ if }|\hat A|=1\\
0\text{ otherwise}\end{array}\right .\text{ for all }\hat A\in 2^N_i\text ,\\
q^{B}_j(t)&=&\left\{\begin{array}{c}1\text{ if }B=A\backslash i\\
0\text{ otherwise}\end{array}\right .\text{ for all }j\in A\backslash i,B\in 2^N_j\text ,\\
q^{\hat B}_{j'}(t)&=&q^{\hat B}_{j'}(t-1)\text{ for all }j'\in A^c,\hat B\in 2^N_{j'}\text .
\end{eqnarray*}
\textsl{Output:} Set $\textbf q^*=\textbf q(t)$.\smallskip

Both \textsc{RoundUp} and \textsc{LocalSearch} yield a sequence
$\textbf q(0),\ldots ,\textbf q(t^*)=\textbf q^*$ where $q_i^*\in ex(\Delta_i)$ for all $i\in N$. In the former at the end of each iteration $t$ the novel
$\textbf q(t)\in\mathcal N(\textbf q(t-1))$ is in the neighborhood of its predecessor. In the latter $\textbf q(t)\not\in\mathcal N(\textbf q(t-1))$ in general, as
in $|P|\leq n$ iterations of \textsl{Loop 1} a partition $\{A^*(1),\ldots ,A^*(|P|)\}=P$ is generated. Selected subsets or blocks $A^*(t)\in 2^N$, $t=1,\ldots ,|P|$ are any of those where
the sum over members $i\in A^*(t)$ of $(i,A^*(t))$-derivatives $\partial W(\textbf q(t-1))/\partial q^{A^*(t)}_i(t-1)$ is maximal. Once a block $A^*(t)$ is selected,
then lines (c) and (d) make all elements $j\in N\backslash A^*(t)$ redistribute the entire membership mass currently placed on subsets $A'\in 2^N_j$ with non-empty intersection
$A'\cap A^*(t)\neq\emptyset$ over those remaining $A\in 2^N_j$ such that, conversely, $A\cap A^*(t)=\emptyset$. The redistribution is such that each of these latter gets a fraction
$w(A)/\sum_{B\in 2^N_j:B\cap A^*(t)=\emptyset}w(B)$ of the newly freed membership mass $\sum_{A'\in 2^N_j:A'\cap A^*(t)\neq\emptyset}q_j^{A'}(t-1)$.  
The subsequent \textsl{Loop 2} checks whether the partition generated by \textsl{Loop 1} may be improved by exctracting some elements from existing blocks and putting them in
singleton blocks of the final output. In the limit, set function $w$ may be such that for some element $i\in N$ global worth decreases when the element joins any
subset $A\in 2^N_i,|A|>1$, that is to say
$w(A)-w(A\backslash i)-w(\{i\})=\sum_{B\in 2^A\backslash 2^{A\backslash i}:|B|>1}\mu^w(B)<0$.
\begin{proposition}
\textsc{LocalSearch}$(W,\textbf q)$ outputs a local maximizer $\textbf q^*$.
\end{proposition}
\begin{proof}
It is plain that the output is a partition $P$ or, with the notation of corollary 8 above, $\textbf q^*=\textbf p$. Accordingly, any element $i\in N$ is
either in a singleton block $\{i\}\in P$ or else in a block $A\in P,i\in A$ such that $|A|>1$. In the former case, any membership reallocation deviating from $p_i^{\{i\}}=1$,
given memberships $p_j,j\in N\backslash i$, yields a cover (fuzzy or not) where global worth is the same as at $\textbf p$, because $\prod_{j\in B\backslash i}p_j^B=0$ for all
$B\in 2_i^N\backslash A$ (see example 2 above). In the latter case, any membership reallocation $q_i$ deviating from $p_i^A=1$ (given memberhips $p_j,j\in N\backslash i$) yields
a cover which is best seen by distinguishing between $2_i^N\backslash A$ and $A$. Also recall that $w(A)-w(A\backslash i)=\sum_{B\in 2^A\backslash 2^{A\backslash i}}\mu^w(B)$.
Again, all membership mass $\sum_{B\in 2_i^N\backslash A}q_i^B>0$ simply collapses on singleton $\{i\}$ because $\prod_{j\in B\backslash i}p_j^B=0$ for all $B\in 2_i^N\backslash A$.
Hence $W(\textbf p)-W(q_i|\textbf p_{-i})=$
\begin{equation*}
=w(A)-w(\{i\})-\left(q_i^A\sum_{B\in 2^A\backslash 2^{A\backslash i}:|B|>1}\mu^w(B)+\sum_{B'\in 2^{A\backslash i}}\mu^w(B')\right)=
\end{equation*}
\begin{equation*}
=\left(p_i^A-q_i^A\right)\sum_{B\in 2^A\backslash 2^{A\backslash i}:|B|>1}\mu^w(B)\text .
\end{equation*}
Now assume that \textbf q is \textit{not} a local maximizer, i.e. $W(\textbf p)-W(q_i|\textbf p_{-i})<0$. Since $p_i^A-q_i^A>0$ (because $p_i^A=1$ and $q_i\in\Delta_i$ is a deviation
from $p_i$), then
\begin{equation*}
\sum_{B\in 2^A\backslash 2^{A\backslash i}:|B|>1}\mu^w(B)=w(A)-w(A\backslash i)-w(\{i\})<0\text .
\end{equation*}
Hence \textbf q cannot be the output of \textsl{Second Loop}.
\end{proof}

In local search methods, the chosen initial canditate solution determines what neighborhoods shall be visited. The range of the objective function in a neighborhood is a set of
real values. In a neighborhood $\mathcal N(\textbf p)$ of a $\textbf p\in\Delta_N$ or partition $P$ only those $\sum_{A\in P:|A|>1}|A|$ elements $i\in A$ in non-sigleton
blocks $A\in P$, $|A|>1$ can modify global worth by reallocating their membership. In view of (the proof of) proposition 13, the only admissible variations obtain by deviating from
$p_i^A=1$ with an alternative membership distribution $q_i^A\in[0,1)$, with $W(q_i|\textbf p_{-i})-W(\textbf p)$ equal to
$(q_i^A-1)\sum_{B\in 2^A\backslash 2^{A\backslash i},|B|>1}\mu^w(B)+(1-q_i^A)w(\{i\})$. Hence, choosing partitions as initial candidate
solutions of \textsc{LocalSearch} is evidently poor. A sensible choice should conversely allow the search to explore different neighborhoods where the objective function may range
widely. A simplest example of such an initial candidate solution is $q_i^A=2^{1-n}$ for all $A\in 2^N_i$ and all $i\in N$, i.e. the  uniform distribution.
On the other hand, the input of local search algorithms
is commonly desired to be close to a global optimum, i.e. a maximizer in the present setting. This translates here into the idea of defining the input by means of set
function $w$. In this view, consider $q_i^A=w(A)/\sum_{B\in 2^N_i}w(B)$, yielding
$\frac{q_i^A}{q_i^B}=\frac{w(A)}{w(B)}=\frac{q^A_j}{q^B_j}$ for all $A,B\in 2^N_i\cap 2^N_j$ and all $i,j\in N$ (see lines (c), (d)).

With a suitable initial candidate solution, the search may be restricted to explore only a maximum number of fuzzy partitions, thereby containing the computational burden.
In particular, if $\textbf q(0)$ is the finest partition $\{\{1\},\ldots ,\{n\}\}$ or $q_i^{\{i\}}(0)=1$ for all $i\in N$, then the
search does not explore any neighborhood at all, and such an input coincides with the output. More reasonably, let $\mathbb A_{\textbf q}^{\max}=\{A_1,\ldots ,A_k\}$ denote the
collection of $\supseteq$-maximal subsets where input memberships are strictly positive. That is, $q_i^{A_{k'}}>0$ for all $i\in A_{k'},1\leq k'\leq k$ as well as $q_j^B=0$ for all
$B\in 2^N\backslash\left(2^{A_1}\cup\cdots\cup2^{A_k}\right)$ and all $j\in B$. Then, the output shall be a partition $P$ each of whose blocks $A\in P$ satisfies $A\subseteq A_{k'}$
for some $1\leq k'\leq k$. Hence, by suitably choosing the input \textbf q, \textsc{LocalSearch} outputs a partition with no less than
some maximum desired number $k(\textbf q)$ blocks.

\section{Lower-dimensional case}
If $\mathcal F\subset 2^N$, then $2^N_i\cap\mathcal F\neq\emptyset$ for every $i\in N$, otherwise the problem reduces to packing the proper subset
$N\backslash\{i:\mathcal F\cap 2^N_i=\emptyset\}$ of elements contained in at least one feasible subset. As outlined in section 1, without additional notation simply let
$\{\emptyset\}\in\mathcal F\ni\{i\}$ for all $i\in N$ with null weigths $w(\emptyset)=0=w(\{i\})$ if $\{i\}\notin\mathcal F$. Thus $(\mathcal F,\supseteq)$ is a poset (partially
ordered set) with bottom element $\emptyset$, and poset function $w:\mathcal F\rightarrow \mathbb R_+$ has its M\"obius inversion $\mu^w:\mathcal F\rightarrow \mathbb R$ \cite{Rota64}.
Memberships $q_i$ distribute over  $\mathcal F_i=2^N_i\cap\mathcal F=\{A_1,\ldots ,A_{|\mathcal F_i|}\}$, with lower($|\mathcal F_i|$)-dimensional unit simplices
\begin{equation*}
\bar\Delta_i=\left\{\left(q_i^{A_1},\ldots ,q_i^{A_{|\mathcal F_i|}}\right)\in\mathbb R^{|\mathcal F_i|}_+: \sum_{1\leq k\leq A_{|\mathcal F_i|}}q_i^{A_k}=1\right\}
\end{equation*}
and corresponding fuzzy covers $\textbf q\in\bar\Delta_N=\times_{1\leq i\leq n}$ $\bar\Delta_i$. Note that a fuzzy cover now may maximally consist of $|\mathcal F|-1$ points in the unit
$n$-dimensional hypercube $[0,1]^n$. Hypercube $[0,1]^n$ is replaced with $\mathcal C(\mathcal F)=co(\{\chi_A:A\in\mathcal F\})\subseteq[0,1]^n$, i.e. the convex hull
of feasible characteristic functions, regarded as $n$-vectors \cite{Branko2001}. Recursively (with $w(\emptyset)=0$), M\"obius inversion $\mu^w:\mathcal F\rightarrow\mathbb R$ is
\begin{equation*}
\mu^w(A)=w(A)-\sum_{B\in\mathcal F:B\subset A}\mu^w(B)\text ,
\end{equation*}
while the MLE $f^w:\mathcal C(\mathcal F)\rightarrow\mathbb R$ of $w$ is
\begin{equation*}
f^w(q^A)=\sum_{B\in\mathcal F\cap 2^A}\left(\prod_{i\in B}q_i^A\right)\mu^w(B)\text . 
\end{equation*}
Therefore, every fuzzy cover $\textbf q\in\bar\Delta_N$ has global worth
\begin{equation*}
W(\textbf q)=\sum_{A\in\mathcal F}\sum_{B\in\mathcal F\cap 2^A}\left(\prod_{i\in B}q_i^A\right)\mu^w(B)\text .
\end{equation*}
For all $i\in N,q_i\in\bar\Delta_i$, and $\textbf q_{-i}\in\bar\Delta_{N\backslash i}=\underset{j\in N\backslash i}{\times}\bar\Delta_j$
\begin{equation*}
W_{-i}(\textbf q_{-i})=\sum_{A\in\mathcal F_i}\left[\sum_{B\in\mathcal F\cap 2^{A\backslash i}}\left(\prod_{j\in B}q^A_j\right)\mu^w(B)\right]+
\end{equation*}
\begin{equation*}
+\sum_{A'\in\mathcal F\backslash\mathcal F_i}\left[\sum_{B'\in\mathcal F\cap 2^{A'}}\left(\prod_{j'\in B'}q^{A'}_{j'}\right)\mu^w(B')\right]\text ,
\end{equation*}
\begin{equation*}
W_i(q_i|\textbf q_{-i})=\sum_{A\in\mathcal F_i}q^A_i\left[\sum_{B\in\mathcal F_i\cap 2^A}\left(\prod_{j\in B\backslash i}q^A_j\right)\mu^w(B)\right]\text ,
\end{equation*}
yielding again
\begin{equation}
W(\textbf q)=W_i(q_i|\textbf q_{-i})+W_{-i}(\textbf q_{-i})\text .
\end{equation}
From (4) above, $w_{\textbf q_{-i}}:\mathcal F_i\rightarrow\mathbb R$ now is
\begin{equation}
w_{\textbf q_{-i}}(A)=\sum_{B\in\mathcal F_i\cap 2^A}\left(\prod_{j\in B\backslash i}q^A_j\right)\mu^w(B)
\end{equation}
for all $i\in N$, all $A\in\mathcal F_i$ and all $\textbf q_{-i}\in\bar\Delta_{N\backslash i}$. 

For each $i\in N$, denote by $ex(\bar\Delta_i)$ the set of $|\mathcal F_i|$ extreme points of simplex $\bar\Delta_i$. Like in the full-dimensional case, at any fuzzy cover
$\hat{\textbf q}\in\bar\Delta_N$ every $i\in N$ such that $\hat q_i\not\in ex(\bar\Delta_i)$ may deviate by concentrating its whole membership on some $A\in\mathcal F_i$
such that $w_{\hat{\textbf q}_{-i}}(A)\geq w_{\hat{\textbf q}_{-i}}(B)$ for all $B\in\mathcal F_i$. This yields a non-decreasing variation
$W(q_i|\hat{\textbf q}_{-i})\geq W(\hat{\textbf q})$ in global worth, with $q_i\in ex(\bar\Delta_i)$. When all $n$ elements do so, one after the other while updating
$w_{\textbf q_{-i}(t)}$ as in \textsc{RoundUp} above, i.e. $t=0,1,\ldots$, then eventually $\textbf q=(q_1,\ldots ,q_n)$ is such that
$\textbf q\in\underset{i\in N}{\times}ex(\bar\Delta_i)$. Yet cases $\mathcal F\subset 2^N$ and $\mathcal F=2^N$ are different in terms of exactness. Specifically, consider any
$\emptyset\neq A\in\mathcal F$ with $|\{i:q_i^A=1\}|\not\in\{0,|A|\}$ or $A_{\textbf q}^+=\{i:q_i^A=1\}\subset A$, with
$f^w(q^A)=\sum_{B\in\mathcal F\cap 2^{A_{\textbf q}^+}}\mu^w(B)$. Then, $\mathcal F\cap 2^{A_{\textbf q}^+}$ is likely to admit no shrinking (see above) yielding
an exact fuzzy cover with same global worth as (non-exact) $\textbf q$.
\begin{proposition}
The values taken on exact fuzzy covers do not saturate the range of $W:\bar\Delta_N\rightarrow\mathbb R_+$.
\end{proposition}
\begin{proof}
By example: $N=\{1,2,3,4\}$ and $\mathcal F=\{N,\{4\},\{1,2\},\{1,3\},\{2,3\}\}$, with worth $w(N)=3$, $w(\{4\})=2$,
$w(\{i,j\})=1$ for $1\leq i<j\leq 3$. Now define $\textbf q=(q_1,\ldots ,q_4)$ by
$q^{\{4\}}_4=1=q^N_i,i=1,2,3$, with non-exactness due to inequality $|\{i:q_i^N>0\}|=3<4=|N|$. As
\begin{equation*}
W(\textbf q)=w(\{4\})+\sum_{1\leq i<j\leq 3}w(\{i,j\})=2+1+1+1
\end{equation*}
and $A^+_{\textbf q}=\{1,2,3\}$, for all distributions $\hat q_1,\hat q_2,\hat q_3$ placing membership only over feasible $B\in\mathcal F\cap 2^{A_{\textbf q}^+}$ global worth is
$W(\hat q_1,\hat q_2,\hat q_3,q_4)<W(\textbf q)$.
\end{proof}

This simple result is useful because it indicates that an inaccurate search among optimal fuzzy covers may lead to a maximizer, either global or local, which is not reducible to any feasible
solution of the original set packing problem. In the present setting, such feasible solutions are partitions $P$ all of whose singleton blocks with worth 0 are not included in the packing.
In fact, similarly to the full-dimensional case, fairly simple conditions may be shown to be sufficient for a partition to be a local maximizer.
\begin{definition}
Any $\hat q_i|\hat{\textbf q}_{-i}=\hat{\textbf q}\in\bar\Delta_N$ is a local maximizer of $W:\bar\Delta_N\rightarrow\mathbb R_+$ if
$W_i(\hat q_i|\hat{\textbf q}_{-i})\geq W_i(q_i|\hat{\textbf q}_{-i})$ for all $i\in N$ and all $q_i\in\bar\Delta_i$ (see (6) above).
\end{definition}
The neighborhood $\mathcal N(\textbf q)\subset\bar\Delta_N$ of $\textbf q\in\bar\Delta_N$ thus is
\begin{equation*}
\mathcal N(\hat{\textbf q})=\underset{i\in N}{\bigcup}\Big\{\textbf q:\textbf q=q_i|\hat{\textbf q}_{-i},q_i\in\bar\Delta_i\Big\}\text .
\end{equation*}
Evidently, there are many partitions $P$ with associated \textbf p such that $\textbf p\in\bar\Delta_N$ (with the notation of corollary 8 above). For example $P_{\bot}=\{\{1\},\ldots ,\{n\}\}$, i.e.
the bottom element of the geometric lattice $(\mathcal P^N, \wedge ,\vee)$ of partitions of $N$ ordered by coarsening \cite{Aigner79}. Other simple examples are given, for every $A\in\mathcal F$, by
the corresponding modular element \cite{Stanley1971} $\{A\}\cup P^{A^c}_{\bot}$ of $(\mathcal P^N, \wedge ,\vee)$ whose unique non-singleton block is $A$. In fact, any partition $P$ such that
$A\in\mathcal F$ for each block $A\in P$ has associated \textbf p satisfying $\textbf p\in\bar\Delta_N$.

\begin{proposition}
Any partition $P$ with associated \textbf p such that $\textbf p\in\bar\Delta_N$ is a local maximizer if for all $A\in P$
\begin{equation*}
w(A)\geq w(\{i\})+\sum_{\hat B\in\mathcal F\cap 2^{A\backslash i}}\mu^w(\hat B)\text .
\end{equation*}
\end{proposition}
\begin{proof}
Firstly note that for all blocks $A\in P$, if any, such that $|A|=1$ there is nothing to prove, as the summation reduces to $w(\emptyset)=0$, and thus there only remains
$w(\{i\})\geq w(\{i\})$. Accordingly, let $A\in P$ and $|A|>1$. For every $i\in A$, any membership reallocation $q_i\in\bar\Delta_i$ deviating from $p_i$ (i.e. $p_i^A=1$), given memberships
$\textbf p_{-i}$ of other elements $j\in N\backslash i$ (i.e. $\bar\Delta_j\ni p_j^{A'}=1$ for all $A'\in P$ and all $j\in A'$), yields $\textbf q=(q_i|\textbf p_{-i})\in\bar\Delta_N$ which is
best analyzed by distinguishing between $\mathcal F_i\backslash A$ and $A$. In particular,
\begin{equation*}
w(A)=w(\{i\})+\sum_{\underset{|B|>1}{B\in\mathcal F_i\cap 2^A}}\mu^w(B)+\sum_{\hat B\in\mathcal F\cap 2^{A\backslash i}}\mu^w(\hat B)\text .
\end{equation*}
All membership mass $\sum_{B\in\mathcal F_i\backslash A}q_i^B>0$ collapses on singleton $\{i\}$, because $\underset{i'\in B\backslash i}{\prod}p_{i'}^B=0$ for all $B\in\mathcal F_i\backslash A$
by the definition of $\textbf p_{-i}$ (see example 2 above). Thus,
\begin{equation*}
W(\textbf p)-W(q_i|\textbf p_{-i})=w(A)-w(\{i\})+
\end{equation*}
\begin{equation*}
-\left(q_i^A\sum_{\underset{|B|>1}{B\in\mathcal F_i\cap 2^A}}\mu^w(B)+\sum_{\hat B\in\mathcal F\cap 2^{A\backslash i}}\mu^w(\hat B)\right)=
\end{equation*}
\begin{equation*}
=\left(p_i^A-q_i^A\right)\sum_{\underset{|B|>1}{B\in\mathcal F_i\cap 2^A}}\mu^w(B)\text .
\end{equation*}
Now assume that $\textbf p$ is \textit{not} a local maximizer, i.e. $W(\textbf p)-W(q_i|\textbf p_{-i})<0$. Since $p_i^A-q_i^A>0$ because $p_i^A=1$
and $q_i\in\bar\Delta_i$ is a deviation from $p_i$, then
\begin{equation*}
\sum_{\underset{|B|>1}{B\in\mathcal F_i\cap 2^A}}\mu^w(B)=w(A)-w(\{i\})-\sum_{\hat B\in\mathcal F\cap 2^{A\backslash i}}\mu^w(\hat B)<0
\end{equation*}
must hold. This contradicts the premise $w(A)\geq w(\{i\})+\sum_{\hat B\in\mathcal F\cap 2^{A\backslash i}}\mu^w(\hat B)$ for all $A\in P$ and $i\in A$, thus completing the proof.
\end{proof}

\section{Local search with cost}
In order to design a gradient-based local search for this lower-dimensional case, the only tool still missing is the derivative, which clearly shall reproduce definition 11 above with
$\mathcal F_i$ in place of $2^N_i$. Before that, as outlined in section 1, let $c:\mathcal F\rightarrow\mathbb N$ count the number $c(A)=|\{B:B\in\mathcal F,B\cap A\neq\emptyset\}|$ of feasible
subsets with which each $A\in\mathcal F$ has non-empty intersection, itself included, i.e. $c(A)\in\{1,\ldots ,|\mathcal F|\}$ is the cost of including $A$ in the packing.
Accordingly, the underlying poset function $\hat w:\mathcal F\rightarrow\mathbb R_+$ now used (still taking positive values only) incorporates both weights $w(A),A\in\mathcal F$ (used thus far)
and costs by means of ratio $\hat w(A)=\frac{w(A)}{c(A)}$. The result is quasi-objective function $\hat W:\bar\Delta_N\rightarrow\mathbb R_+$, obtained via MLE
$f^{\hat w}:\mathcal C(\mathcal F)\rightarrow\mathbb R_+$ of $\hat w$, i.e. $\hat W(\textbf q)=\sum_{A\in\mathcal F}f^{\hat w}(q^A)$, and all of the above applies invariately simply replacing
$W$ with $\hat W$.
\begin{definition}
The $(i,A)$-derivative of $\hat W$ at $\textbf q\in\bar\Delta_N$,
\begin{equation*}
\text{$A\in\mathcal F_i$, is }\partial\hat W(\textbf q)/\partial q^A_i=\hat W(\overline{\textbf q}(i,A))-\hat W(\underline{\textbf q}(i,A))=
\end{equation*}
\begin{equation*}
=\hat W_i\Big(\overline q_i(i,A)|\overline{\textbf q}_{-i}(i,A)\Big)-\hat W_i\Big(\underline q_i(i,A)|\underline{\textbf q}_{-i}(i,A)\Big)\text ,
\end{equation*}
with $\overline{\textbf q}(i,A)=\Big(\overline q_1(i,A),\ldots ,\overline q_n(i,A)\Big)$ given by
\begin{equation*}
\overline q_j^B(i,A)=\left\{\begin{array}{c} q_j^B\text{ for all }j\in N\backslash i,B\in\mathcal F_j \\ 1\text{ for }j=i,B=A\\0\text{ for }j=i,B\neq A \end{array}\right.
\text ,
\end{equation*}
and $\underline{\textbf q}(i,A)=\Big(\underline q_1(i,A),\ldots ,\underline q_n(i,A)\Big)$ given by
\begin{equation*}
\underline q_j^B(i,A)=\left\{\begin{array}{c} q_j^B\text{ for all }j\in N\backslash i,B\in\mathcal F_j \\ 0\text{ for }j=i\text{ and all }B\in\mathcal F_i \end{array}\right.
\text .
\end{equation*}
The (full) gradient of $\hat W$ at $\textbf q\in\bar\Delta_N$ is
\begin{equation*}
\nabla\hat W(\textbf q)=\left\{\partial\hat W(\textbf q)/\partial q^A_i:i\in N,A\in\mathcal F_i\right\}\in\mathbb R^{\sum_{i\in N}|\mathcal F_i|}
\end{equation*}
as well as the
$i$-gradient $\nabla_i\hat W(\textbf q)\in\mathbb R^{|\mathcal F_i|}$ of $\hat W$ at $\textbf q=(q_i|\textbf q_{-i})\in\bar\Delta_N$ is poset function
$\nabla_i\hat W(\textbf q):\mathcal F_i\rightarrow\mathbb R$ defined by
$\nabla_i\hat W(\textbf q)(A)=\hat w_{\textbf q_{-i}}(A)$ for all $A\in\mathcal F_i$, where $\hat w_{\textbf q_{-i}}$ is given by (7) with $\hat w$ in place of $w$. Again, membership distribution
$\underline q_i(i,A)$ is the null one: its $|\mathcal F_i|$ entries are all 0, hence $\underline q_i(i,A)\not\in\bar\Delta_i$.
\end{definition}
\smallskip

\textsc{LS-WithCost}$(\hat w,\textbf q)$\smallskip

\textsl{Initialize:} Set $t=0$ and $\textbf q(0)=\textbf q$, with requirement $|\{i:q_i^A>0\}|\in\{0,|A|\}$ for all $A\in\mathcal F,\hat w(A)>0$.\smallskip

\textsl{Loop 1:} While $0<\sum_{i\in A}q^A_i(t)<|A|$ for a $A\in\mathcal F$, set $t=t+1$ and:
\begin{enumerate}
\item[(a)] select a $A^*(t)\in\mathcal F$ such that
\begin{equation*}
\min_{i\in A^*(t)}\text{ }\hat w_{\textbf q_{-i}(t-1)}(A^*(t))\geq\min_{j\in B}\text{ }\hat w_{\textbf q_{-j}(t-1)}(B)
\end{equation*}
for all $B\in 2^N$ such that $0<\sum_{i\in B}q^B_j(t)<|B|$,
\item[(b)] for $i\in A^*(t)$ and $A\in \mathcal F_i$, define
\begin{equation*}
q_i^A(t)=
\left\{\begin{array}{c}1\text{ if }A=A^*(t)\text ,\\
0\text{ if }A\neq A^*(t)\text ,\end{array}\right.
\end{equation*}
\item[(c)] for $j\in N\backslash A^*(t)$ and $A\in\mathcal F_j$ with $A\cap A^*(t)=\emptyset$, define $q^A_j(t)=q_j^A(t-1)+$
\begin{equation*}
+\left(\hat w(A)\sum_{\underset{B\cap A^*(t)\neq\emptyset}{B\in\mathcal F_j}}q_j^B(t-1)\right)
\left(\sum_{\underset{B'\cap A^*(t)=\emptyset}{B'\in\mathcal F_j}}\hat w(B')\right)^{-1}
\end{equation*}
\item[(d)] for $j\in N\backslash A^*(t)$ and $A\in\mathcal F_j$ with $A\cap A^*(t)\neq\emptyset$, define
\begin{equation*}
q^A_j(t)=0\text .
\end{equation*}
\item[(e)] for $A\in\mathcal F$ with $A\cap A^*(t)=\emptyset$, update cost function by
\begin{equation*}
c(A)=|\{B:B\in\mathcal F,B\cap A\neq\emptyset=B\cap A^*(t)\}|
\end{equation*}
and plug it into $\hat w$.
\end{enumerate}
\textsl{Loop 2:} While $q_i^A(t)=1,|A|>1$ for a $i\in N$ and
\begin{equation*}
w(A)<w(\{i\})+\sum_{\hat B\in\mathcal F\cap 2^{A\backslash i}}\mu^w(\hat B)\text ,
\end{equation*}
set $t=t+1$ and define:
\begin{eqnarray*}
q^{\hat A}_i(t)&=&\left\{\begin{array}{c}1\text{ if }|\hat A|=1\\
0\text{ otherwise}\end{array}\right .\text{ for all }\hat A\in\mathcal F_i\text ,\\
q^{B}_j(t)&=&\left\{\begin{array}{c}1\text{ if }B=A\backslash i\\
0\text{ otherwise}\end{array}\right .\text{ for all }j\in A\backslash i,B\in\mathcal F_j\text ,\\
q^{\hat B}_{j'}(t)&=&q^{\hat B}_{j'}(t-1)\text{ for all }j'\in A^c,\hat B\in\mathcal F_{j'}\text .
\end{eqnarray*}
\textsl{Output:} Set $\textbf q^*=\textbf q(t)$.\smallskip

Both \textsc{LocalSearch} and \textsc{LS-WithCost} generate in $|P|\leq n$ iterations of \textsl{Loop 1} a partition
$\{A^*(1),\ldots ,A^*(|P|)\}=P$. Now selected blocks $A^*(t)\in\mathcal F$, $1\leq t\leq |P|$ are any of those feasible subsets where the minimum over elements
$i\in A^*(t)$ of $(i,A^*(t))$-derivatives $\partial\hat W(\textbf q(t-1))/\partial q^{A^*(t)}_i(t-1)$ is maximal. The following \textsl{Loop 2} again checks whether the partition
generated by \textsl{Loop 1} may be improved by exctracting some elements from existing blocks and putting them in singleton blocks of the final output, which thus allows for the following.
\begin{proposition}
\textsc{LS-WithCost}$(W,\textbf q)$ outputs a local maximizer $\textbf q^*$.
\end{proposition}
\begin{proof}
Follows from proposition 16 since \textsl{Loop 2} deals with $w$, not with $\hat w$.
\end{proof}

Concerning input $\textbf q=\textbf q(0)$, consider again setting
\begin{equation*}
q_i^A=\frac{\hat w(A)}{\sum_{B\in\mathcal F_i}\hat w(B)}\text{ for all }A\in\mathcal F_i,i\in N\text{, which entails}
\end{equation*}
\begin{equation*}
\frac{q_i^A}{q_i^B}=\frac{w(A)c(B)}{w(B)c(A)}=\frac{q_j^A}{q_j^B}\text{ for all }A,B\in\mathcal F_i\cap\mathcal F_j,i,j\in N\text .
\end{equation*}

Evidently, \textsl{Loop 1} may take exactly the same form as in \textsc{LocalSearch}, that is with selected blocks $A^*(t)\in\mathcal F,t=1,\ldots ,|P|$ of the generated partition $P$ being any
of those feasible subsets where the sum, rather than the minimum, over elements $i\in A^*(t)$ of $(i,A^*(t))$-derivatives $\partial\hat W(\textbf q(t-1))/\partial q^{A^*(t)}_i(t-1)$ is
maximal. This possibility seems appropriate in applicative scenarios, where set packing is mostly dealt with in its weighted version. Yet, using the minimum in place of the sum, although
computationally more demanding, appears interesting for $k$-uniform set packing problems (see section 1), widely studied in computational complexity. In fact, for the $k$-uniform case M\"obius
inversion is $\mu^{\hat w}(A)=\frac{1}{c(A)}$ if $|A|=k$ and $\mu^{\hat w}(A)=0$ if $|A|\in\{0,1\}$ for all $A\in\mathcal F$ (recall the convention $\{\emptyset\}\in\mathcal F\ni\{i\}$ for all $i\in N$),
with the cost function updated at each iteration according to line (e). It is also evident that in $k$-uniform set packing \textsl{Loop 2} is ineffective.

\section{Near-Boolean functions}
Boolean functions \cite{BooleanFunctions} provide key analytical tools and methods with a variety of important applications. Beyond set packing problems that here constitute the main benchmark,
this section further develops the full-dimensional case detailed in section 2 with the aim to indicate additional opportunities obtained from expanding the standard framework where pseudo-Boolean models
are traditionally exploited. Recall that Boolean functions of $n$ variables have form $f:\{0,1\}^n\rightarrow\{0,1\}$, and constitute a subclass of pseudo-Boolean functions
$f:\{0,1\}^n\rightarrow\mathbb R$, which in turn admit the unique MLE $\hat f:[0,1]^n\rightarrow\mathbb R$ over the whole $n$-dimensional unit hypercube extensively employed thus far. The $n$ variables
thus range each in the unit interval $[0,1]$. Such a setting is here expanded by letting each variable $i=1,\ldots ,n$ range in a $2^{n-1}-1$-dimensional simplex $\Delta_i$, with the goal to
evaluate collections of fuzzy subsets of a $n$-set through the MLE given by (1) and (2).
\begin{definition}
Near-Boolean functions of $n$ variables have form
\begin{equation*}
F:\underset{1\leq i\leq n}{\times}ex(\Delta_i)\rightarrow\mathbb R\text .
\end{equation*}
\end{definition}
Following \cite[p. 4]{HammerHolzman}, denote by $N=\{1,\ldots ,n\}$ the set of indices of variables (i.e., the ground set in previous sections). As already observed, any
pseudo-Boolean function has a unique expression as a multilinear polynomial $f$ in $n$ variables: $f(x_1,\ldots ,x_n)=\underset{A\subseteq N}{\sum}\left(\alpha_A\underset{i\in A}{\prod}x_i\right)$.
Uniqueness is customarily shown by induction on the size $0\leq |A|\leq n$ of subsets of variables \cite[p. 162]{BorosHammer02}, although it seems a by-product of the M\"obius inversion of (po)set
functions \cite{Rota64}, which is unique indeed, i.e. $\alpha_A,A\in 2^N$ simply is the M\"obius inversion of some unique set function $w:2^N\rightarrow\mathbb R$ such that $w(A)=f(\chi_A)$, where
$\chi_A$ is the characteristic function defined in section 2, i.e. $\chi:2^N\rightarrow\{0,1\}^n$ with $\chi(A)=\chi_A$.
\begin{definition}
The MLE $\hat F$ of near-Boolean functions $F$ has polynomial form
\begin{equation*}
\hat F:\underset{1\leq i\leq n}{\times}\Delta_i\rightarrow\mathbb R
\end{equation*}
given by expression (2) in section 2, that is
\begin{equation*}
\hat F(\textbf q)=\sum_{A\in 2^N}\left[\sum_{B\supseteq A}\left(\prod_{i\in A}q_i^B\right)\right ]\mu^w(A)\text ,
\end{equation*}
 with (see above) $\textbf q=(q_1,\ldots ,q_n)$ and $q_i=(q_i^{A_1},\ldots ,q_i^{A_{2^{n-1}}})\in\Delta_i$.
\end{definition}
\subsection{$k$-degree approximations}
In line with \cite{HammerHolzman}, the issue of approximating a given near-Boolean function $F$ by means of the least squares criterion amounts to determine a near-Boolean function $F_k$ such that
\begin{equation}
\sum_{\textbf q\in\underset{i\in N}{\times}ex(\Delta_i)}\left[F(\textbf q)-F_k(\textbf q)\right]^2
\end{equation}
attains its minimum over all near-Boolean functions $F_k$ with polynomial MLE $\hat F_k$ of degree $k$, that is
\begin{equation*}
\hat F_k(\textbf q)=\sum_{\underset{|A|\leq k}{A\in 2^N}}\left[\sum_{B\supseteq A}\left(\prod_{i\in A}q_i^B\right)\right ]\mu^w(A)\text ,
\end{equation*}
or, equivalently stated in terms of the underlying set function $w$, such that $\mu^w(A)=0$ if $|A|>k$.

Near-Boolean functions $F$ take their values on $n$-product $\underset{i\in N}{\times}ex(\Delta_i)$, and $|ex(\Delta_i)|=2^{n-1}$ for each $i\in N$. They might thus be regarded as points
$F\in\mathbb R^{2^{n(n-1)}}$ in a $2^{n(n-1)}$-dimensional vector space. In view of proposition 4 concerning exactness, this seems conceptually incorrect and with useless enumerative demand. Specifically,
for every partition $P\in\mathcal P^N$ with associated $\textbf p\in\underset{i\in N}{\times}ex(\Delta_i)$, there clearly exist many non-exact $\textbf q\in \underset{i\in N}{\times}ex(\Delta_i)$ such that
$F(\textbf q)=F(\textbf p)$ (see corollary 8 above). Counting these redundant extreme points of simplices appears wothless. For this reason, $k$-degree approximation is dealt with by replacing expression (8)
with the following, applying to partitions \textbf p only
\begin{equation}
\sum_{\textbf p\in\underset{i\in N}{\times}ex(\Delta_i)}\left[F(\textbf p)-F_k(\textbf p)\right]^2\text .
\end{equation}
The number $|\mathcal P^N|$ of partitions of a $n$-set is given by \textit{Bell number} $\mathcal B_n$ \cite{Rota1964,Aigner79}. Accordingly, near-Boolean functions might be regarded as points
$F\in\mathbb R^{\mathcal B_n}$ in a $\mathcal B_n$-dimensional vector space. Still, this also is far too large, as points in such a vector space correspond in fact to generic partition functions, i.e. with
M\"obius inversion free to live on every partition $P\in\mathcal P^N$. Conversely, near-Boolean functions factually involve only partition functions $h:\mathcal P^N\rightarrow\mathbb R$ such that
$h(P)=h_w(P)=\sum_{A\in P}w(A)$ for some set function $w:2^N\rightarrow\mathbb R$. The M\"obius inversion of these partition functions lives only on the $2^n-n$ modular elements \cite{Stanley1971} of lattice
$(\mathcal P^N, \wedge ,\vee)$, namely on those partitions with a number of non-sigleton blocks $\leq 1$. In turn, this entails that when regarded as points in a vector space (i.e. expressed as a
linear combination of a basis, see above) these functions may be seen as $h_w\in\mathbb R^{2^n-n}$. This can be shown via recursion through the M\"obius inversion of \textit{additively separable}
partition functions; it is here omitted being contained in \cite{GilboaLehrer90GG,GilboaLehrer91VI}.

When all these facts are properly taken into account, the issue of $k$-degree approximation for near-Boolean functions is seen to reduce to the same issue for traditional pseudo-Boolean
functions, which is already exhaustively detailed in \cite{HammerHolzman}. What is crucial emphasizing though, is that while for pseudo-Boolean functions there exists a unique best $k$-degree
approximation for all $0\leq k\leq n$, on the other hand every near-Boolean function admits a continuum of set functions $w$ determining their unique best $k$-degree approximation, and this applies
to all $0<k\leq n$. Furthermore, case $k=0$ is of no concern in
that the emptyset cannot be a block of any partition. In particular, consider the linear case, which is the main case in \cite[p. 4]{HammerHolzman}. The issue is to find a best (least squares) approximation
$F_1$ of any given $F$. That is, the set function $w$ determining $F_1$ has to satisfy $w(A)=\sum_{i\in A}w(\{i\})$ for all $A\in 2^N$. Then,
\begin{equation*}
h_w(P)=\sum_{A\in P}w(A)=\sum_{A\in P}\sum_{i\in A}w(\{i\})=w(N)
\end{equation*}
for all $P\in\mathcal P^N$. Thus, $h_w$ is a constant partition function, i.e. a valuation \cite{Aigner79} of partition lattice $(\mathcal P^N, \wedge ,\vee)$. Also, any further linear
$v:2^N\rightarrow\mathbb R$ such that $v(N)=w(N)$ satisfies $h_v(P)=h_w(P)$ for all $P\in\mathcal P^N$. This means that there is a continuum (i.e. a non-unit $n-1$-dimensional simplex) of equivalent linear
$v\neq w$ such that $h_w=h_v$, obtained each by distributing arbitrarily the whole of $w(N)$ over the $n$ singletons $\{i\}\in 2^N$. Cases $k>1$ are still all the same. To see this, consider a set function
$w$ such that $\mu^w(A)\neq 0$ for one or more (possibly all $\binom{n}{k}$) subsets $A\in 2^N$ such that $|A|=k$. Now fix arbitrarily $n$ values $v(\{i\}),i\in N$ with
$\sum_{i\in N}w(\{i\})=\sum_{i\in N}v(\{i\})$. For all $A\in2^N,|A|>1$ M\"obius inversion $\mu^v:2^N\rightarrow\mathbb R$ can always be determined uniquely through recursion by 
\begin{equation*}
v(A)+\sum_{i\in A^c}v(\{i\})=\sum_{B\subseteq A}\mu^v(B)+\sum_{i\in A^c}v(\{i\})=
\end{equation*}
\begin{equation*}
=w(A)+\sum_{i\in A^c}w(\{i\})=\sum_{B\subseteq A}\mu^w(B)+\sum_{i\in A^c}w(\{i\})\text .
\end{equation*}
One thing must be absolutely clear: there is a continuum of equivalent set functions (i.e. $w$ and $v$) available for the sought $k$-degree approximation $F_k$, but still the $\mathcal B_n$ values
taken by $F_k$ are unique and indpendent from the chosen set function in the continuum available (and thus such $\mathcal B_n$ values are also unique for any fixed $F$ to be approximated). Any $F$ clearly is
its own best $n$-degree approximation.
\begin{remark}
There is a continuum of set functions $w$ equivalently determining the MLE $\hat F$ of $F$. 
\end{remark}
\subsection{Near-Boolean games}
In view of the above definition of local maximizers relying on equilibrium conditions for strategic $n$-player games, and having mentioned additive separablity of partition functions or global games
\cite{GilboaLehrer90GG}, it seems now natural to  consider variables as players in near-Boolean games (see also \cite[section 3]{HammerHolzman}).
\begin{definition}
A near-Boolean $n$-player game is a triple $(N,F,\pi)$ such that $N=\{1,\ldots ,n\}$ is the player set and $F$ is a near-Boolean function taking real values on
profiles $\textbf q=(q_1,\ldots ,q_n)\in\underset{i\in N}{\times}ex(\Delta_i)$ of strategies, while payoffs $\pi:\underset{i\in N}{\times}ex(\Delta_i)\rightarrow\mathbb R^n$ are efficient, namely
$\pi(\textbf q)=(\pi_1(\textbf q),\ldots ,\pi_n(\textbf q))$ satisfies $\sum_{i\in N}\pi_i(\textbf q)=F(\textbf q)$ at all $\textbf q\in\underset{i\in N}{\times}ex(\Delta_i)$.
\end{definition}
\begin{definition}
A fuzzy near-Boolean $n$-player game is a triple $(N,\hat F,\pi)$ such that $N=\{1,\ldots ,n\}$ is the player set and $\hat F$ is the MLE of a near-Boolean function taking
real values on strategy profiles $\textbf q=(q_1,\ldots ,q_n)\in\underset{i\in N}{\times}\Delta_i$, while $\pi:\underset{i\in N}{\times}\Delta_i\rightarrow\mathbb R^n$ efficiently assigns payoffs
$\pi(\textbf q)=(\pi_1(\textbf q),\ldots ,\pi_n(\textbf q))$ to players, i.e. $\sum_{i\in N}\pi_i(\textbf q)=\hat F(\textbf q)$ at all $\textbf q\in\underset{i\in N}{\times}\Delta_i$.
\end{definition}
Game theory since many years is mostly concerned with finite sets of players, which is indeed the case both for near-Boolean games and fuzzy ones. Given this, a main distinction is between games where players
have either finite or else infinite sets of strategies, with near-Boolean games in the former class and fuzzy ones in the latter. In addition, players may play either deterministic (i.e. pure) or else
randomized (i.e. mixed) strategies. In the latter case equilibrium conditions are stated in terms of expected payoffs, and by means of fixed point arguments for upper hemicontinuous correspondences such
conditions are commonly fulfilled \cite[p. 260]{Micro}. The sets of deterministic strategies in fuzzy near-Boolean games are precisely the sets of randomized strategies in near-Boolean games. Nevertheless, the
payoffs for the fuzzy setting evidently are \textit{not} expectations.   

The main framework where (possibly fuzzy) near-Boolean games seem appropriate is coalition formation, which combines both strategic and cooperative games. A generic strategy profile
$\textbf q\in\underset{i\in N}{\times}ex(\Delta_i)$ of near-Boolean (non-fuzzy) games may well fail to be exact (see proposition 4), but it is understood at this point that there is a
unique partition $P$ of $N$ with associated $\textbf p\in\underset{i\in N}{\times}ex(\Delta_i)$ such that $F(\textbf p)=F(\textbf q)$. Let $\textbf p(\textbf q)$ be such a unique \textbf p.
In view of the above discussion on approximations, it is also clear that for every $\textbf p$ there are many \textbf q such that
\textbf p=\textbf p(\textbf q). In these terms, near-Boolean games model stategic coalition formation in a very handy manner, in that they totally by-pass the need to define a mechanism mapping
strategy profiles into partitions (or coalition structures) of players \cite{Slikker2001436} . More precisely, a mechanism is a mapping $M:\underset{i\in N}{\times}ex(\Delta_i)\rightarrow\mathcal P^N$
such that when each player $i\in N$ specifies a coalition $A_i\in 2^N_i$, then $M(A_1,\ldots ,A_n)=P$ is a resulting partition. If the $n$ specified coalitions $A_i,i\in N$ are such that for some partition
$P$ it holds $A_i=A$ for all $i\in A$ and all $A\in P$, then $M(A_1,\ldots ,A_n)=P$. Otherwise, the partition $P'=M(A_1,\ldots ,A_n)$ generated by the mechanism shall be a rather fine one, i.e. possibly
consisting of many small blocks (depending on the chosen mechanism). Conversely, near-Boolean games do not need any mechanism, in that even if players' strategies $(q_1,\ldots ,q_n)=\textbf q$ are such that
\textbf q does not correspond to a partition, still the global worth $F(\textbf q)$ is that attained at the partition $P$ with corresponding \textbf p(\textbf q), i.e. whose blocks $A\in P$ each include maximal
subsets of players choosing the same $A'\supseteq A$. 

Now \textit{fix} a set function or coalitional game $v:2^N\rightarrow\mathbb R_+,v(\emptyset)=0$ such that $F(\textbf p)=\sum_{A\in P}v(A)$ for all partitions $P$, and let the payoffs
be defined, for all $i\in A$ and all $A\in P$, by
\begin{equation*}
\pi_i(\textbf q)=\sum_{B\in 2^A\backslash 2^{A\backslash i}}\frac{\mu^{w}(B)}{|B|}\text ,
\end{equation*}
where $P$ is the partition with associated \textbf p(\textbf q). This is in fact a well-known coalition formation game, where payoffs are given by the \textit{Shapley value} \cite{Roth88}.
\begin{definition}
A local maximizer $\textbf q\in\underset{i\in N}{\times}ex(\Delta_i)$ of near-Boolean function $F$ satisfies for all $i\in N$ and all $q'_i\in ex(\Delta_i)$ inequality
$F(\textbf q)\geq F(q_i'|\textbf q_{-i})$.
\end{definition}
\begin{remark}
If payoffs are given by $\pi_i(\textbf q)=\frac{\omega_iF(\textbf q)}{\sum_{j\in N}\omega_j}$ for all $i\in N$, with $\omega_1,\ldots,\omega_n>0$, then near Boolean games are (pure) common interest
potential games \cite{MondererShapley96,Bowles}. The set of equilibria of near-Boolean game $(N,F,\pi)$ coincides with the set of local maximizers of $F$.
\end{remark}
\section{Conclusions}
Via polynomial MLE, the proposed near-Boolean functions of $n$ variables take values on the $n$-product of high-dimensional unit simplices.
This enables to approach discrete optimization problems, namely set packing, with an objective function defined over a continuous domain, with feasible solutions found at extreme points of the simplices.
Approximations with polynomials of bounded degree are finally discussed, while near-Boolean $n$-player games provide a new modeling of strategic coalition formation.  

\bibliographystyle{abbrv}
\bibliography{biblioContinuousSetPacking}

\end{document}